\RequirePackage{luatex85}

\documentclass[11pt,a4paper]{article}
\usepackage{iftex}

\ifPDFTeX
\usepackage[utf8]{inputenc}
\usepackage[T1]{fontenc}
\fi

\usepackage{amsmath}
\usepackage{amssymb}
\usepackage{amsthm}
\usepackage{mathtools}
\usepackage{thmtools}
\usepackage{isomath}
\usepackage{algorithm}

\ifPDFTeX
\usepackage{libertine}
\usepackage[libertine]{newtxmath} % must load after ams packages
\usepackage[scaled=0.96]{zi4} % use Inconsolata as the typewriter font
\else
\usepackage{fontspec}
\usepackage{unicode-math}
\fi

\usepackage[DIV=11]{typearea} % large value = less margin, recommended values: 10pt: 8, 11pt: 10, 12pt: 12

\usepackage{microtype}

\usepackage[algo2e,ruled,vlined,linesnumbered]{algorithm2e}

\usepackage{xcolor}

\usepackage[disable]{todonotes}
\usepackage{lineno}

\usepackage[
	style=alphabetic,
    maxalphanames=4,
    minalphanames=4,
    maxcitenames=4,
    minnames=1,
    maxbibnames=20,
	backref=true,
	doi=true,
	url=true,
	backend=biber
]{biblatex}

\usepackage[ocgcolorlinks]{hyperref} % Option ocgcolorlinks makes links only colored on screen and not on printouts

\newcommand*\samethanks[1][\value{footnote}]{\footnotemark[#1]}

\allowdisplaybreaks[1]

\makeatletter
\g@addto@macro\bfseries{\boldmath}
\makeatother

\makeatletter
\g@addto@macro\mdseries{\unboldmath}
\g@addto@macro\normalfont{\unboldmath}
\g@addto@macro\rmfamily{\unboldmath}
\g@addto@macro\upshape{\unboldmath}
\makeatother

\DeclareCiteCommand{\citem}
    {}
    {\mkbibbrackets{\bibhyperref{\usebibmacro{postnote}}}}
    {\multicitedelim}
    {}

\renewcommand*{\multicitedelim}{\addcomma\space}

\AtEveryCitekey{\clearfield{note}}

\makeatletter
\AtBeginDocument{%
    \newlength{\temp@x}%
    \newlength{\temp@y}%
    \newlength{\temp@w}%
    \newlength{\temp@h}%
    \def\my@coords#1#2#3#4{%
      \setlength{\temp@x}{#1}%
      \setlength{\temp@y}{#2}%
      \setlength{\temp@w}{#3}%
      \setlength{\temp@h}{#4}%
      \adjustlengths{}%
      \my@pdfliteral{\strip@pt\temp@x\space\strip@pt\temp@y\space\strip@pt\temp@w\space\strip@pt\temp@h\space re}}%
    \ifpdf
      \typeout{In PDF mode}%
      \def\my@pdfliteral#1{\pdfliteral page{#1}}% I don't know why % this command...
      \def\adjustlengths{}%
    \fi
    \ifxetex
      \def\my@pdfliteral #1{}% isn't equivalent to this one
      \def\adjustlengths{\setlength{\temp@h}{-\temp@h}\addtolength{\temp@y}{1in}\addtolength{\temp@x}{-1in}}%
    \fi%
    \def\Hy@colorlink#1{%
      \begingroup
        \ifHy@ocgcolorlinks
          \def\Hy@ocgcolor{#1}%
          \my@pdfliteral{q}%
          \my@pdfliteral{7 Tr}% Set text mode to clipping-only
        \else
          \HyColor@UseColor#1%
        \fi
    }%
    \def\Hy@endcolorlink{%
      \ifHy@ocgcolorlinks%
        \my@pdfliteral{/OC/OCPrint BDC}%
        \my@coords{0pt}{0pt}{\pdfpagewidth}{\pdfpageheight}%
        \my@pdfliteral{F}% Fill clipping path (the url's text) with
        \my@pdfliteral{EMC/OC/OCView BDC}%
        \begingroup%
          \expandafter\HyColor@UseColor\Hy@ocgcolor%
          \my@coords{0pt}{0pt}{\pdfpagewidth}{\pdfpageheight}%
          \my@pdfliteral{F}% Fill clipping path (the url's text)
        \endgroup%
        \my@pdfliteral{EMC}%
        \my@pdfliteral{0 Tr}% Reset text to normal mode
        \my@pdfliteral{Q}%
      \fi
      \endgroup
    }%
}
\makeatother

\newcommand{\antonis}[1]{\todo[linecolor=orange!50!black,backgroundcolor=orange!25,bordercolor=orange!50!black]{\scriptsize \textbf{AS:} #1}}

\colorlet{DarkRed}{red!50!black}
\colorlet{DarkGreen}{green!50!black}
\colorlet{DarkBlue}{blue!50!black}

\hypersetup{
	linkcolor = DarkRed,
	citecolor = DarkGreen,
	urlcolor = DarkBlue,
	bookmarksnumbered = true,
	linktocpage = true
}

\DontPrintSemicolon
\SetKwProg{Procedure}{Procedure}{}{}
\SetProcNameSty{textsc}
\SetFuncSty{textsc}
\SetKw{KwAnd}{and}
\SetKw{KwDo}{do}

\SetCommentSty{mycommfont}
\SetKwFunction{AliceWins}{AliceWins}

\declaretheorem[numberwithin=section]{theorem}
\declaretheorem[numberlike=theorem]{lemma}

\declaretheorem[numberlike=theorem]{observation}

\DeclareMathOperator{\poly}{poly}

\DeclareMathOperator{\Z}{\mathbb Z}

\title{Fast Algorithms for Energy Games in Special Cases}

\author{
  Sebastian Forster\thanks{Department of Computer Science, University of Salzburg, Austria. This work is supported by the Austrian Science Fund (FWF): P 32863-N. This project has received funding from the European Research Council (ERC) under the European Union's Horizon 2020 research and innovation programme (grant agreement No~947702).}
  \and
  Antonis Skarlatos\samethanks
  \and
  Tijn de Vos \samethanks
}

\date{}

\hypersetup{
	pdftitle = {Fast Algorithms for Energy Games in Special Cases},
	pdfauthor = {Sebastian Forster, Antonis Skarlatos, Tijn de Vos}
}

\begin{document}
\maketitle
\begin{abstract}
In this paper, we study algorithms for special cases of energy games, a class of turn-based games on graphs that show up in the quantitative analysis of reactive systems. In an energy game, the vertices of a weighted directed graph belong either to Alice or to Bob. A token is moved to a next vertex by the player controlling its current location, and its energy is changed by the weight of the edge. 
Given a fixed starting vertex and initial energy, Alice wins the game if the energy of the token remains nonnegative at every moment. 
If the energy goes below zero at some point, then Bob wins. The problem of determining the winner in an energy game lies in $\mathsf{NP} \cap \mathsf{coNP}$.
It is a long standing open problem whether a polynomial time algorithm for this problem exists. 

We devise new algorithms for three special cases of the problem.
The first two results focus on the single-player version, where either Alice or Bob controls the whole game graph.
We develop an $\tilde{O}(n^\omega W^\omega)$ time algorithm for a game graph controlled by Alice, by providing a reduction to the All-Pairs 
Nonnegative Prefix Paths problem (APNP), where $W$ is
the maximum absolute value of any edge weight
and $\omega$ is the best exponent for matrix multiplication. Thus we study the APNP problem separately, for which we develop an $\tilde{O}(n^\omega W^\omega)$ time algorithm. For both problems, we improve over the state of the art of $\tilde O(mn)$ for small $W$. %\antonis{Maybe provide the region defined by m, n, W, that we improve?}
For the APNP problem, we also provide a conditional lower bound which states that there is no $O(n^{3-\epsilon})$ time algorithm for any $\epsilon > 0$, unless the APSP Hypothesis fails.
For a game graph controlled by Bob, we obtain a near-linear time algorithm. 
Regarding our third result, we present a variant of the value iteration algorithm, and we prove that it gives an $O(mn)$ time 
algorithm for game graphs without negative cycles,
which improves a previous upper bound. The all-Bob algorithm is randomized, all other algorithms are deterministic.
\end{abstract}

% \newpage
% \tableofcontents
% \newpage

\newpage

\section{Introduction}
%\tijn{explain $\tilde O$}
Energy games belong to a class of turn-based games on graphs that show up in the quantitative analysis of reactive systems.
A game graph can possibly represent a scheduling problem, where vertices are the configurations of the system
and edges carry positive or negative values representing the evolution of resources. Thus, in this model
resources can be consumed or produced. 
The energy games problem has been introduced in the early 2000s~\cite{ChakrabartiAHS03,BouyerFLMS08}, but also has been implicitly studied before due to its ties to mean-payoff games~\cite{Gurvich88}.
Energy games have applications in, among others, computer aided verification and automata theory~\cite{ChakrabartiAHS03, BloemCHJ09,CernyCHRS11}, and in online and streaming problems~\cite{ZwickP96}.
From its computational perspective, 
the problem of determining the winner in an energy game lies in $ \mathsf{NP} \cap \mathsf{coNP} $.
It is an intriguing open problem whether a polynomial time algorithm for this problem exists.

An energy game is played by two players, say Alice and Bob, on a \emph{game graph}, which is a weighted directed graph such that each 
vertex is either controlled by Alice or Bob.
The game starts by placing a token with an initial energy on a starting vertex. The game is played in rounds, and 
every time the token is located at a vertex controlled by Alice, then Alice chooses the next 
location of the token among the outgoing edges, otherwise Bob chooses the next move.
The token has an \emph{energy level} (in the beginning this is equal to the initial energy)
and every time it traverses an edge, the weight
of the edge is added to the energy level (a negative weight amounts to a reduction of the energy level).
The objectives of the players are as follows: Alice wants to minimize the initial energy that is necessary to keep the energy level nonnegative at all times, whereas Bob wants to maximize this value (and possibly drive it to $ \infty$).
The computational problem is to determine for each vertex the minimum initial energy such that Alice can guarantee against all choices of Bob that the energy level always stays nonnegative. 

Energy games are a generalization of parity games~\cite{Jurdzinski98,BrimCDGR11}, polynomial-time equivalent to mean-payoff games~\cite{BouyerFLMS08,BrimCDGR11}, 
and a special case of simple stochastic games~\cite{ZwickP96}.
Recent progress on parity games yielded several quasipolynomial time algorithms~\cite{CaludeJKLS22}, but the corresponding techniques seem to not carry over to energy and mean-payoff games~\cite{FijalkowGO20}. 
Consequently, the complexity of energy games is still ``stuck'' at pseudopolynomial~\cite{BrimCDGR11} or subexponential time~\cite{BjorklundV07}.
Hence, in this paper we focus on interesting special cases (which are non-trivial problems) that admit fast algorithms.
Two of these cases are game graphs where all vertices are controlled by one player, and the third case are game graphs with no negative
cycles. 

\paragraph{All-Pairs Nonnegative Prefix Paths.}
We also study another reachability problem with energy constraints \cite{HMR19, dorfman2023optimal}, the \emph{All-Pairs Nonnegative Prefix Paths (APNP) problem}. 
In this problem, the goal is to determine for
every pair of vertices whether there exists a path $\pi$ such that the weight of each prefix of $\pi$ is nonnegative. 
We use this problem to obtain
the result for the special case where Alice controls the whole game graph,
since the two problems are closely related.
Dorfman, Kaplan, Tarjan, and Zwick~\cite{dorfman2023optimal} solve the more general problem, where for each pair of vertices the goal is to find the path $\pi$ of maximum weight among all options. This problem naturally generalizes APSP, and they solve it in $O(mn+n^2\log n)$ time. 
%\antonis{Why does it generalize APSP? I think we should define APSP here?}

\paragraph{Energy Games.}
%\tijn{two more approaches to consider: (1) "An Optimal Strategy
%  Improvement Algorithm for Solving Parity and Payoff Games" 2008, where a
%  symmetric strategy improvement algorithm for parity games has been developed,
%  which can also be immediately applied to solve mean-payoff games; (2) "Solving
%  Mean-Payoff Games via Quasi Dominions" 2020, where an algorithm structurally
%  close to the approach of [20] has been studied.
%}
The state-of-the-art algorithms for the energy games are either deterministic with running time 
$ O\left(\min (mnW, mn2^{n/2} \log W)\right) 
$~\cite{BrimCDGR11,DorfmanKZ19} or randomized with subexponential running time $ 2^{O(\sqrt{n\log n})}$~\cite{BjorklundV07}. Special cases of the 
energy games have been studied by Chatterjee, Henzinger, Krinninger, and Nanongkai~\cite{ChatterjeeHKN14}.
They present a variant of the value iteration algorithm of~\cite{BrimCDGR11} with running time $O(m|A|)$, where $A$ is a sorted list containing all possible minimum
energy values. This does not improve the general case, as $A$ in the worst case is the list $\{0, 1, \dots, nW, \infty\}$. However, it does give a faster running time if the weights adhere to certain restrictions. Moreover, they develop a scaling algorithm with running time 
$O(mn\log W (\log \frac{W}{P} + 1) + mn \frac{W}{P}) $, where $ P \in \{ \frac{1}{n}, \ldots, W \} $ is a lower bound on the \emph{penalty} of the game. 

For the special case where there are no negative cycles in the game graph, the penalty can
be set to $W$, and the scaling algorithms of \cite{ChatterjeeHKN14} solves the problem in $ O(mn \log W)$ time.
For another special case where the whole game graph is controlled by Alice, 
Brim and Chaloupka \cite{brim2012using} provided an $\tilde{O}(mn)$\footnote{We write $\tilde{O}(f)$ for $O(f \poly\log f)$.} running time algorithm as a subroutine for the two-players version.

\paragraph{Mean-Payoff Games.}
In a mean-payoff game, the objective of Alice is to maximize the average of the weights of edges traversed so far, whereas Bob's objective is to minimize this mean payoff.
It is well known that any energy games instance can be solved by solving $O(n\log (nW))$ mean-payoff games~\cite{BouyerFLMS08}, and any mean-payoff game instance can be solved by solving $ O (\log (nW)) $ energy games with maximal weight~$nW$~\cite{BrimCDGR11}\footnote{Unless stated otherwise, we always consider the versions of the games where we compute the mean-payoff value/minimum initial energy for \emph{all} vertices.}.
Thus, any of the aforementioned algorithms for solving energy games also yields an algorithm for solving mean-payoff games at the expense of at most an additional factor of $ O (n \log (nW)) $ in the running time. Zwick and Paterson~\cite{ZwickP96} provided the first
pseudopolynomial time algorithm that computes all the mean-payoff values, with $O(mn^3W)$ running time. Later, the running time was improved by
Brim, Chaloupka, Doyen, Gentiline, and Raskin~\cite{BrimCDGR11} to $O(mn^2W\log(nW))$, using their reduction to energy games.
The state-of-the-art algorithm for solving a mean-payoff game is due to Comin and Rizzi~\cite{CominR17} which runs in $O(mn^2W)$ time.

\subsection{Our Results and Techniques}
\paragraph{All-Pairs Nonnegative Prefix Paths.}
The version of All-Pairs Nonnegative Prefix Paths (APNP) problem where we want to find the path of maximum weight~\cite{dorfman2023optimal}, naturally generalizes the 
All-Pairs Shortest Paths (APSP) problem. The APSP Hypothesis states that there is no $O(n^{3-\epsilon})$ time algorithm for the APSP, for any $\epsilon > 0$.
However, this version of APNP is more than what is necessary for the application of energy games. We show that the weaker version which only computes reachability (as APNP has been defined), 
also does not allow for a $O(n^{3-\epsilon})$ time algorithm
for any $\epsilon > 0$, under the APSP Hypothesis. 
%\antonis{Let's review this paragraph one more time. Write more about the hypothesis?}

\begin{restatable}{theorem}{ThmAPNPhard}
\label{th:APNPhard}
    Unless the APSP Hypothesis fails, there is no $O(n^{3-\epsilon})$ time algorithm
    that solves the All-Pairs Nonnegative Prefix Paths problem, for any $\epsilon > 0$.
\end{restatable} 

We parameterize the maximum absolute value of
any edge weight $W$, and we obtain an algorithm with a faster
running time for small values of $W$.

\begin{restatable}{theorem}{ThmAPNP}
\label{th:apnpW}
    There exists a deterministic algorithm that, given a graph $G = (V, E, w)$ with edge weights in the interval $[-W, W]$,  solves the All-Pairs Nonnegative Prefix Paths problem in $\tilde{O}(n^\omega W^\omega)$ time.
\end{restatable}

\paragraph{All-Alice.}
Our first contribution regarding the special cases of energy games,
concerns the \emph{all-Alice case} in which all vertices are controlled by Alice. Note that if we fix a strategy for Bob in any game graph, this can be seen as an all-Alice instance. 
\begin{restatable}{theorem}{ThmAllAlice}\label{thm:allalice}
There exists a deterministic algorithm that, given a game graph $G=(V,E,w)$ in which all vertices are controlled by Alice, computes the minimum sufficient energy of all vertices in $\tilde{O}(n^\omega W^\omega)$ time.
\end{restatable}
Note that the aforementioned reduction from energy games to mean-payoff games always introduces Bob vertices.
Thus, algorithms for the all-Alice mean-payoff decision problem cannot be leveraged by this reduction to compute the minimum energies in the all-Alice case.
\antonis{should we mention one more time here the $\tilde{O}(mn)$ alg of Brim and Chaloupka?}
Our approach for the all-Alice case consists of two steps. In the first step, we identify all vertices $Z$ such that minimum initial energy $0$ suffices, by using Theorem~\ref{th:apnpW}. 
In the second step, we compute the paths of least energy reaching any vertex in $Z$.
For small values of $W$, this improves on the state-of-the-art $\tilde O(mn)$ algorithm~\cite{brim2012using}.

%%%%%%%%%%%%%%%%%%%%%%%%%%%%%%%%%%%%%%%%

\paragraph{All-Bob.}
Our second contribution regarding the special cases
of energy games, is a faster algorithm for the \emph{all-Bob case} in which all vertices are controlled by Bob. Note that if we fix a strategy for Alice in any game graph, this can be seen as an all-Bob instance. 
\begin{restatable}{theorem}{ThmAllBob}\label{thm:allbob}
There exists a randomized (Las Vegas) algorithm that, given a game graph $G=(V,E,w)$ in which all vertices are controlled by Bob, computes the minimum sufficient energy of all vertices, and with high probability the algorithm takes $O (m\log^2 n \log nW\log\log n) $ time.
\end{restatable}

To the best of our knowledge, the fastest known algorithm for the all-Bob case is implied by the reduction to the mean-payoff decision problem and has a running time of $ \tilde O (m n \log^2 W) $.
This comes from $ \tilde O (n \log W) $ calls to the state-of-the-art negative-cycle detection algorithm~\cite{BernsteinNWN22,bringmann2023negative}.

Our approach for the all-Bob case consists of two steps.
In the first step, we run a negative-cycle detection algorithm to remove all vertices reaching a negative cycle.
In the second step, we add an artificial sink to the graph with an edge from every vertex to the sink, and we compute the shortest path of every vertex to the sink using a single-source shortest paths (SSSP) algorithm.
Note that this construction is very close to Johnson's method for computing suitable vertex potentials~\cite{Johnson77}. 
Further note that, since energy games are not symmetric for Alice and Bob, our near-linear time all-Bob algorithm has no implications for the all-Alice case. 

%%%%%%%%%%%%%%%%%%%%%%%%%%%%%%%%%%%%%%%%

\paragraph{No Negative Cycles.}
Finally, we give an improved algorithm for the special case where there are no negative cycles.
\begin{restatable}{theorem}{ThmNoNegCycles}\label{thm:nonnegcycles}
There exists a deterministic algorithm that, given a grame graph $G=(V,E)$ without negative cycles, computes the minimum sufficient energy of all vertices in $ O (m n) $ time.
\end{restatable}
\antonis{Should we say it's "strongly-polynomial"? If it actually is.}
To the best of our knowledge, the fastest known algorithm for this special case has a running time of $  O (m n \log W) $ by running the above mentioned algorithm of Chatterjee, Henzinger, Krinninger, and Nanongkai~\cite{ChatterjeeHKN14} with penalty $ P = W $. 
We use a new variant of the value iteration algorithm where the energy function after $i$ steps corresponds to the minimum energy function in an \emph{$i$-round game}. 
A similar variant has been used by Chatterjee, Doyen, Randour, and Raskin~\cite{Chatterjee0RR15} for the Mean-Payoff games. We adapt this algorithm and provide the necessary analysis to use it for energy games. 

An $i$-round game is a finite version of the energy game, where a token is passed for only $i$~rounds. In this version, the goal is to find the initial energy that Alice needs, in order to keep the energy level nonnegative for these $i$-rounds. Then we show that in graphs without negative cycle, the infinite game is equivalent to the $n$-round game. 

\paragraph{Structure of the paper.} In the next section, we provide some preliminaries, including the formal definition of an energy game. In Section~\ref{sec:APNP}, we study the All-Pairs Nonnegative Prefix Paths
problem, and we present an algorithm for the special case that the edge weights
are in $\{-1, 0, +1\}$, an algorithm for general edge weights, and a lower bound. Next in Section~\ref{sec:all-alice}, we consider the all-Alice case
by reducing this problem to the All-Pairs Nonnegative Prefix Paths problem.
In Section~\ref{sec:all-bob}, we consider the all-Bob case, and finally in Section~\ref{sc:nonegcycles}, we consider game graphs without negative cycles.

\section{Preliminaries}\label{sc:prelim}
\paragraph{Graphs.}
Given a directed graph $G=(V, E, w)$, we denote by $n = |V|$ the number
of vertices, by $m = |E|$ the number of edges, and by $W$ the maximum
absolute value of any edge weight. Also, we denote $N^+(v)$ for the \emph{out-neighborhood} of $v$, i.e., $N^+(v) := \{ u\in V : (v,u)\in E\}$. Further, we denote $\deg^+(v)$ for the \emph{out-degree} of $v$, i.e., $\deg^+(v) := |N^+(v)|$. Similarly, $N^-(v)$ and $\deg^-(v)$ denote the \emph{in-neighborhood} and \emph{in-degree} respectively. 

A \emph{path} $P$ is a sequence of vertices $u_0u_1\cdots$ such that $(u_i,u_{i+1})\in E$ for every $i\geq 0$. We say a path is \emph{finite} if it contains a finite number of vertices (counted with multiplicity). We say a path is \emph{simple} if each vertex appears at most once. A \emph{lasso} is a path of the form $u_0u_1\cdots u_ju_i$, where the vertices $u_0,\dots, u_j$ are disjoint and $i < j$. In other words, it is a simple path leading to a cycle. A \emph{nonnegative prefix path} is a path $P=u_0u_1\cdots$ such that $\sum_{j=0}^{i-1} w(u_j,u_{j+1})\geq 0$ for all $1\leq i \leq |P|$. Further, we denote the weight of a path $P=u_0u_1\cdots$ by $w(P):=\sum_{j=0}^{|P|-1} w(u_j,u_{j+1})$.
For a fixed path $P = u_0u_1\cdots$, the energy level $e(u_i)$ of a vertex $u_i$ in $P$ is equal to $\sum_{j=0}^{i-1} w(u_j,u_{j+1})$. That is,
the sum of all the weights along $P$ until $u_i$.

Let $G=(V, E, w)$ be a directed graph with edge weights $-1$ and $+1$, and let $s, t \in V$ be two vertices of $G$.
Then, a \emph{Dyck path from $s$ to $t$} is a nonnegative prefix path from $s$ to $t$ of total weight zero \cite{Bradford17}. 
%\tijn{To the best of our knowledge, Dyck paths are not usually applied to this
%  context. It seems that the first work in which Dyck paths have been used for
%  digraphs (with the required meaning) is [8]. I suggest adding a citation for
%  Dyck paths and the connections between their different semantics.}
For a graph $H$, we refer to the corresponding
functions by using $H$ as subscript (e.g., we use the notation  $w_H(\cdot)$
for the weight function of $H$).

\paragraph{Energy Games.}
An energy game is an infinite duration game played by two players, Alice and Bob. The game is played on a \emph{game graph} which a weighted directed graph $G = (V, E, w)$, where each vertex has at least one outgoing edge. The weights are integers and lie in the range $\{-W,-W+1, \dots, W-1,W\}$. The set of vertices is partitioned in two sets $V_A$ and $V_B$, controlled by Alice and Bob respectively. Furthermore, we are given a starting vertex $s\in V$, and initial energy $e_0\geq 0$. We start with position $v_0=s$. After the $i_{\text{th}}$ round, we are at a position $v_i \in V$ and have energy $e_i$. In the $i_{\text{th}}$ round, if $v_{i-1} \in V_A$ ($v_{i-1}\in V_B$) then Alice (Bob) chooses a next vertex $v_i\in N^+(v_{i-1})$ and the energy changes to $e_i=e_{i-1} + w(v_{i-1}, v_i)$. The game ends when $e_i < 0$, in which case we say that Bob wins. If the game never ends, namely, $e_i \geq 0$ for all $i \geq 0$, we say that Alice wins. The goal is to find out the minimum initial energy $e_0 \geq 0$ such that Alice wins when both players play optimally. Note that allowing $e_0=\infty$ means that such an energy always exist. 

To make this goal more formal, we have to introduce \emph{strategies}. A strategy for Alice (Bob) tells us given the current point $v_i\in V_A$ ($v_i\in V_B$) and the history of the game, $v_0, \dots, v_i$, where to move next.
It turns out that we can restrict ourselves to \emph{positional strategies}~\cite{ehrenfeuchtM79,BouyerFLMS08}, which are deterministic and do not depend on the history of the game. We denote a positional strategy of Alice by $\sigma\colon V_A \to V$ where $\sigma(v)\in N^+(v)$ for $v \in V_A$, and a positional strategy of Bob by $\tau\colon V_B \to V$ where $\tau(v)\in N^+(v)$ for $v \in V_B$. For any pair of strategies $(\sigma,\tau)$ we define $G(\sigma, \tau)$ to be the subgraph $(V,E')$ corresponding to these strategies, where $E' = \{(v,\sigma(v)) : v \in V_A\}\cup \{(v,\tau(v)) : v \in V_B\}$. Note that in this graph each vertex has exactly one out-neighbor. 
Let $P_i$ be the unique path $s=u_0,u_1,\dots,u_i$ in $G(\sigma,\tau)$ of length $i$ originating at $s$. Then at vertex $s$ with initial energy $e_0$ and with these strategies,
Alice wins if $e_0 + w(P_i) \geq 0$ for all $i \geq 0$, and Bob wins if $e_0 +w(P_i) < 0$ for at least one $i\geq 0$. The \emph{minimum sufficient energy at $s$ with respect to $\sigma$ and $\tau$} is the minimum energy
such that Alice wins, namely $e_{G(\sigma,\tau)}(s):= \max\{0, -\inf_{i\geq0} w(P_i)\}$.  Finally, we define the \emph{minimum sufficient energy} at $s$ as follows: 
$$e^*_G(s) := \min_\sigma \max_\tau e_{G(\sigma,\tau)}(s),$$
where the minimization and the maximization are over all the positional strategies $\sigma$ of Alice and $\tau$ of Bob, respectively. 
We omit the subscript $G$, and use $e_{\sigma,\tau}(s)$ instead of $e_{G(\sigma,\tau)}(s)$, whenever this is clear from the context.
By Martin's determinacy theorem~\cite{Martin75}, we have that $\min_\sigma \max_\tau e_{\sigma,\tau}(s)= \max_\tau \min_\sigma e_{\sigma,\tau}(s)$, thus the outcome is independent of the order in which the players pick their strategy. Now we can define \emph{optimal strategies} as follows. A strategy $\sigma^*$ is an optimal strategy for Alice, if $e_{\sigma^*, \tau}(s) \leq e^*(s)$ for any strategy $\tau$ of Bob. Similarly, $\tau^*$ is an optimal strategy for Bob, if $e_{\sigma,\tau^*}(s) \geq e^*(s)$ for any strategy $\sigma$ of Alice. 
An \emph{energy function} is a function $e\colon V \to \Z_{\geq0}\cup\{\infty\}$. The function $e^*_G(\cdot)$ (or $e^*(\cdot)$) as defined above, is the \emph{minimum sufficient energy function}.

\section{All-Pairs Nonnegative Prefix Paths Problem} \label{sec:APNP}
In this section, we study the All-Pairs Nonnegative Prefix Paths (APNP) problem.
The goal of this problem is to determine for every pair of vertices
whether there exists a nonnegative prefix path between them.
A similar problem is the \emph{All-Pairs Dyck-Reachability problem},
where the goal is to determine for every pair of vertices
whether there exists a Dyck path between them (given that the edge weights are in $\{-1, +1\}$). Furthermore, another
standard problem is the \emph{transitive closure problem}, which asks to determine for every pair of vertices whether there exists a path between them. 

Bradford~\cite{Bradford17} provided an $\tilde{O}(n^\omega)$ time algorithm
for the All-Pairs Dyck-Reachability problem. 
Moreover, the transitive closure problem admits an $\tilde{O}(n^\omega)$ algorithm~\cite{AhoHU74}.

\begin{theorem} \label{th:bradalg}
    There exists a deterministic algorithm that, given a graph $G = (V, E, w)$ with edge weights in $\{-1, 1\}$, solves the All-Pairs Dyck-Reachability problem in $\tilde{O}(n^\omega)$ time.
\end{theorem}

Our approach for the APNP problem consists of two stages. At first,
we solve the APNP problem for the special case where the edge weights are from the set $\{-1, 0, +1\}$,
by exploiting the algorithm of \cite{Bradford17} for the All-Pairs Dyck-Reachability
problem. Afterwards, we extend our algorithm to work with general weights, by showing
that a reduction used in \cite{AlonGM97} preserves the properties we need. 

In the end of the section, we also present a conditional lower bound for the APNP problem under
the APSP Hypothesis, which is one of the main hypotheses in fine-grained complexity.

\subsection{All-Pairs Nonnegative Prefix Paths with edge weights in \texorpdfstring{$\{-1, 0, +1\}$}{-1,0,+1}}
Consider a graph $G = (V, E)$ with edge weights $-1$ and $+1$. By definition, we 
have that any Dyck path is also a nonnegative prefix path. However, the
opposite is not necessarily true. 
Recall that nonnegative prefix paths allow the energy level of their last vertex to be a strictly positive value, while in Dyck paths this value must be zero. This implies that
an All-Pairs Dyck-Reachability algorithm does not trivially gives us 
an All-Pairs Nonnegative Prefix Paths algorithm. Nevertheless, we show how to overcome this issue and
we use an All-Pairs Dyck-Reachability algorithm as a subroutine in order to 
solve the All-Pairs Nonnegative Prefix Paths problem. 

\paragraph{Algorithm for the $\{-1, 0, +1\}$ case.}
Consider a directed graph $G = (V, E, w)$, with edge weights in $\{-1, 0, +1\}$.
In the beginning of the algorithm, we construct a graph $G_2$ as follows.

\begin{enumerate}
\item  Initially, we create a new graph $G_1 = (V_1, E_1, w)$ by replacing every edge of zero weight with an edge of weight $+1$ and an edge of weight $-1$. 
Specifically, for each vertex $u$ with at least one outgoing edge $(u, v) \in E$ with $w(u, v) = 0$, 
we add a new vertex $u'$, and add an edge $(u,u')$ with $w(u,u')=+1$. Next, for each edge $(u,v)\in E$ with $w(u,v)=0$, we remove the edge $(u, v)$, and add the edge $(u', v)$ with weight $-1$.\footnote{Note that by doing the naive thing which is to replace each edge of zero weight by two edges, one with weight $+1$ and one with weight $-1$, potentially blows up the number of vertices to $\Omega(m)$. In turn, since the running
time depends on the number of vertices, this translates to a blow up
of the running time.}

\item Next, we run on $G_1$ the algorithm of Theorem~\ref{th:bradalg}, which solves the All-Pairs Dyck-Reachability in time $\tilde{O}(n^\omega)$ for edge weights in $\{-1, +1\}$.

\item Finally, we create another new graph $G_2 = (V, E_2)$ with the original vertex set and an edge set $E_2$ defined as follows.
The set $E_2$ contains an edge $(u, v) \in V \times V$ if and only if there is a Dyck path from $u$ to $v$ in $G_1$ or 
$w(u, v) = 1$ in $G$ with $(u, v) \in E$. \antonis{Simpler alternative: add a negative self-loop on $t$?}
\end{enumerate}

In the end, we run on $G_2$ a transitive closure algorithm,
and we return that there is a nonnegative prefix path in $G$
if and only if there is a path in $G_2$. Notice that graphs $G$ and $G_1$ are weighted, while $G_2$ is unweighted.

\paragraph{Analysis of the algorithm.}
The following observation shows that the replacement of zero weight edges is valid,
in the sense that nonnegative prefix paths of total weight zero\footnote{Observe that a Dyck path is a nonnegative prefix path of total weight zero consisting only of edges $-1$ and $+1$. Thus, we avoid to use the term \emph{Dyck path} for $G$ because it may contains edges of weight zero.} in $G$ correspond
to Dyck paths in $G_1$ and vice versa. 
Moreover, we prove the claim that the transitive closure problem in $G_2$ is equivalent to
the All-Pairs Nonnegative Prefix Paths problem in $G$.

\begin{observation} \label{obs:GG1Dyck}
    For every pair of vertices $u, v \in V$, there exists a nonnegative prefix path
    of total weight zero from $u$ to $v$ in $G$ if and only if there exists
    a Dyck path from $u$ to $v$ in $G_1$.
\end{observation} 
%\antonis{Proof.}

\begin{lemma} \label{lem:nonegtopath}
    For every pair of vertices $u, v \in V$, there exists a nonnegative prefix path from $u$ to $v$ in $G$ if and only if there exists a path from $u$ to $v$ in $G_2$.    
\end{lemma}
\begin{proof}
    Assume that there exists a nonnegative prefix path $\pi$ from $u$ to $v$ in $G$. Let $a$ be the first vertex after $u$ along $\pi$ with 
    a minimum energy level. Initially, we show that the edge $(u, a)$ appears in $G_2$.
    Since $\pi$ is a nonnegative prefix path, we have that $e(u) \leq e(a)$. If $e(u) < e(a)$, then there must be an edge 
    $(u, a)$ in $G$ with weight $+1$. Also if $e(u) = e(a)$, then the subpath of $\pi$ from $u$ to $a$ is a nonnegative prefix path of total weight zero. Then by Observation~\ref{obs:GG1Dyck}, the subpath of
    $\pi$ from $u$ to $a$ is a Dyck path in $G_1$.
    Therefore, in both cases we have added the edge $(u, a)$ in $G_2$. As the vertex $a$ has a minimum energy level, we can apply
    the same argument iteratively starting from $a$, to conclude that there exists a path from $u$ to $v$ in $G_2$.
    
    Assume now that there exists a path $\pi$ from $u$ to $v$ in $G_2$. 
    By construction, the edges of $\pi$ correspond either to edges in $G$ with weight $+1$ or to Dyck paths in $G_1$.
    By Observation~\ref{obs:GG1Dyck}, these Dyck paths in $G_1$ correspond to nonnegative prefix paths of total weight zero in $G$.
    Since positive edges increase the energy level and nonnegative prefix paths at least maintain the energy level, we conclude that
    there exists a nonnegative prefix path from $u$ to $v$ in $G$.
\end{proof}

\begin{lemma} \label{lm:apnpspecial}
    There exists a deterministic algorithm that, given a graph $G = (V, E, w)$ with edge weights in $\{-1, 0, +1\}$, solves the All-Pairs Nonnegative Prefix Paths problem in $\tilde{O}(n^\omega)$ time.
\end{lemma}
\begin{proof}
    The number of vertices of $G_1$ is $O(n)$ by construction, where $n$ is the initial number of vertices in $G$. Hence, the construction of $G_2$ runs in $\tilde{O}(n^\omega)$. Moreover,
    the transitive closure problem in $G_2$ can be solved in $\tilde{O}(n^\omega)$ time as well~\cite{AhoHU74}.
    Thus by Lemma~\ref{lem:nonegtopath}, the claim follows.
\end{proof}

\subsection{All-Pairs Nonnegative Prefix Paths with general edge weights}
We extend now Lemma~\ref{lm:apnpspecial} for graphs with general edge weights, in the cost of an extra factor $W^\omega$ in the running time.
The idea is to use the reduction by Alon, Galil, and Margalit~\cite{AlonGM97}, who reduce the All-Pairs Shortest Paths (APSP) problem with general edge weights
to the special case where the edge weights are in $\{-1, 0, +1\}$. We present the reduction for completeness, and we prove that the same reduction also
preserves the properties that we need for the All-Pairs Nonnegative Prefix Paths problem.

\paragraph{Reduction from general weights to $\{-1, 0, +1\}$ \cite{AlonGM97}.}
Given a graph $G = (V, E, w)$ with weights in the interval $[-W, W]$, we create another graph $G'$ with weights only in $\{-1, 0, +1\}$, as follows.
For every vertex $v \in V$ in $G$, we create $2W+1$ vertices $\{v^i\}_{i=-W}^W$ in $G'$.
We say that vertex $v^0$ of $G'$ is the origin of vertex $v$.
Then, we add in $G'$ an edge $(v^{i+1}, v^i)$ of weight $-1$, for
every $-W \leq i \leq -1$, and an edge $(v^{i-1}, v^i)$ of weight $1$,
for every $1 \leq i \leq W$.
Moreover, for every edge $(u, v)$ of weight $k$ in $G$, we add an edge $(u^k, v^0)$ of zero weight in $G'$. 

\ThmAPNP*
\begin{proof}
    The idea is to apply the reduction mentioned above and use
    the algorithm of Lemma~\ref{lm:apnpspecial} in $G'$. 
    Then, we claim that there exists a nonnegative prefix path from $u$ to $v$ in $G$ if and only if there exists a nonnegative prefix path from $u^0$ to $v^0$
    in $G'$.
    
    Regarding the running time, since
    the number of vertices of the new graph $G'$ after the reduction 
    becomes $\Theta(nW)$, the running time of the algorithm becomes
    $\tilde{O}((nW)^\omega)$. It remains to prove the correctness
    of the algorithm.

    Let $\pi$ be a nonnegative prefix path from $u$ to $v$ in $G$.
    We construct a path $\pi'$ from $u^0$ to $v^0$ in $G'$ as follows. For every edge $(a, b) \in \pi$ of 
    weight $k$, we add to $\pi'$ the unique subpath from $a^0$ 
    to $b^0$ of weight $k$ in $G'$, which exists by construction.
    Since $\pi$ is a nonnegative prefix path in $G$, and every subpath we add to $\pi'$ 
    consists either only of edges with weight in $\{-1, 0\}$
    or only of edges with weight in $\{0, +1\}$, we can infer that $\pi'$
    is a nonnegative prefix path from $u^0$ to $v^0$ in $G'$.
    
    For the other direction, let $\pi'$ be a nonnegative prefix path
    from $u^0$ to $v^0$ in $G'$. We construct a path $\pi$ from $u$ to $v$ in $G$ as follows. Let $a^0$ be the first vertex after $u^0$ along $\pi'$ such that, $a^0$ is
    the origin vertex of a different vertex than $u$ (i.e., $a^0$ is the origin
    of a vertex $a \neq u$). By construction, there exists an edge $(u, a)$ of
    weight $k$ in $G$, where $k$ is the weight of the subpath from $u^0$ to $a^0$
    in $\pi'$. We add the edge $(u, a)$ in $\pi$, and continue with the construction
    of $\pi$ by applying the same argument iteratively starting from $a^0$ until we reach $v^0$.
    Since $\pi'$ is a nonnegative prefix path in $G'$, and each prefix of $\pi$
    corresponds to a prefix in $\pi'$, we can infer that $\pi$ is a nonnegative
    prefix path from $u$ to $v$ in $G$.

    Therefore, the pair of vertices $\{u^0, v^0\}$ in $G'$ contains the information for the pair of vertices $\{u, v\}$ in $G$,
    and so the claim follows.
\end{proof}

\subsection{Lower bound for All-Pairs Nonnegative Prefix Paths}
We prove a lower bound on the running time of All-Pairs Nonnegative Prefix Paths problem under the APSP Hypothesis.
The APSP Hypothesis is an assertions that the All-Pairs Shortest Paths (APSP) problem cannot be solved in truly
subcubic $O(n^{3-\epsilon})$ time, for any $\epsilon > 0$. Vassilevska Williams and Williams \cite{williams2010subcubic}
proved that APSP and Negative Triangle are equivalent under subcubic reductions. The Negative Triangle problem
is defined as follows. Given a graph $G = (V, E, w)$, the goal is to find three vertices $a, b, c$ such that
$w(a, b) + w(b, c) + w(c, a) < 0$, that is, the vertices $a, b, c$ form a negative weight cycle.

Recently, a reduction from the Negative Triangle problem to the $h\text{-hop-bounded}$ s-t path problem
was given by Polak and Kociumaka~\cite{polak2022bellman}, in order to prove a hardness result for the latter.
Motivated by this reduction, we also reduce the Negative Triangle problem to the All-Pairs Nonnegative Prefix Paths problem to obtain
a hardness result for the All-Pairs Nonnegative Prefix Paths problem, as shown in Theorem~\ref{th:APNPhard}.

We first provide an auxiliary lemma, which we also
use later in Lemma~\ref{lm:Alice0}.

\begin{lemma} \label{lem:nonegcyclehasnonnegpref}
    Given a graph $G = (V, E, w)$, 
    let $C$ be a nonnegative weight cycle in $G$ (i.e., $w(C) \geq 0$).
    Then, there is a vertex $u \in C$ in the cycle, such that there exists a nonnegative prefix path in $G$ from $u$ to itself along $C$.
\end{lemma}
\begin{proof}
    Let $Q \subseteq C$ be a 
    subpath of $C$ with the most negative total weight, and $Q'$ be the rest of $C$ (i.e., $Q \cup Q' = C$). Notice that the weight of all prefixes in $Q'$
    must be nonnegative, otherwise this negative weight prefix could be merged with $Q$, contradicting the fact that $Q$ is the subpath of $C$
    with the most negative total weight. Moreover, as $w(C) \geq 0$ we have that $w(Q') \geq -w(Q)$.
    Since by definition of $Q$, there is no prefix of $Q$ with more negative total weight, it holds that $Q' \cup Q$
    is a nonnegative prefix path from the first vertex of $Q'$ to itself along $C$.
\end{proof}

\ThmAPNPhard*
\begin{proof}
    Consider a Negative Triangle instance $G = (V, E)$. We create a directed graph $G_1=(V_1,E_1)$ as follows.
    The vertex set $V_1$ of $G_1$ consists of five copies of all vertices, i.e., $V_1:= \{ v^i : v\in V, i\in \{1,2,3,4,5\}\} $.
    %For each vertex $v \in V$, let $v^i$ be the corresponding vertex in $G_1$ in the $i_\text{th}$ copy. 
    For every edge $(u, v) \in E$ of weight $w(u, v)$, we add an edge $(u^i, v^{i+1})$ to $E_1$ with weight $-w(u, v)$, for $1 \leq i < 4$. Also for each vertex $v \in V$, we add an edge $(v^4, v^5)$ of weight $w_\text{min} = -1$.
    
    We claim that there exists a negative weight triangle in $G$ if and only
    if there is a vertex $v \in V$ such that there exists a nonnegative prefix path from $v^1$ to $v^5$ in $G_1$.
    In this case, since the reduction is subcubic and the time to check all vertices in $G_1$ is $O(n)$, an $O(n^{3-\epsilon})$ time algorithm
    for the All-Pairs Nonnegative Prefix Paths problem would imply an $O(n^{3-\epsilon})$ time algorithm
    for the Negative Triangle problem, for any $\epsilon > 0$, contradicting the APSP Hypothesis. 
    
    We proceed with the proof of the claim. 
    Suppose that there are three vertices $a, b, c$ that form a negative
    weight cycle $C$ in $G$, and let $G_2$ be the graph $G$ after flipping
    the sign of the weights. Then we have that $w_{G_2}(C) > 0$ in $G_2$,
    and based on Lemma~\ref{lem:nonegcyclehasnonnegpref}
    there is a vertex $v \in C$, such that there exists a nonnegative
    prefix path in $G_2$ from $v$ to itself along $C$.
    Notice that $v$ can be either $a, b$ or $c$, and by
    construction, the paths $a^1b^2c^3a^4a^5$,
    $b^1c^2a^3b^4b^5$, and $c^1a^2b^3c^4c^5$ exist in $G_1$. Thus
    without loss of generality, we can assume that $v$ is $a$
    and we use the path $a^1b^2c^3a^4a^5$ in $G_1$.
    By construction, it holds that $w_{G_1}(a^1, b^2) = w_{G_2}(a, b),
    w_{G_1}(b^2, c^3) = w_{G_2}(b, c)$, $w_{G_1}(c^3, a^4) = w_{G_2}(c, a)$
    and $w_{G_1}(a^4, a^5) = w_\text{min}$.
    The path $abca$ is a nonnegative prefix path
    in $G_2$, and so the path
    $a^1b^2c^3a^4$ is
    a nonnegative prefix path in $G_1$ as well. Moreover since $w_{G_2}(C) > 0$, we have that $w_{G_2}(C) \geq
    -w_\text{min}$, which implies
    that: \[w_{G_1}(a^1, b^2) + w_{G_1}(b^2, c^3) + w_{G_1}(c^3, a^4)
    \geq -w_\text{min}.\]
    Thus, we can conclude that the path 
    $a^1b^2c^3a^4a^5$ is
    a nonnegative prefix path in $G_1$. 

    For the other direction, let $a^1b^2c^3a^4a^5$ be a nonnegative
    prefix path in $G_1$. By construction of $G_1$ and the fact that $G$ does not contain
    self-loops, it must be the case that the corresponding vertices $a, b, c$ must be pairwise
    different in $G$. By definition of a nonnegative
    prefix path, it holds that: \[w_{G_1}(a^1, b^2) + w_{G_1}(b^2, c^3) + w_{G_1}(c^3, a^4) \geq -w_\text{min} > 0.\]
    By construction, we have that $w(a, b) = -w_{G_1}(a^1, b^2),
    w(b, c) = -w_{G_1}(b^2, c^3)$ and $w(c, a) = -w_{G_1}(c^3, a^4)$.
    Therefore, it is true
    that $w(a, b) + w(b, c) + w(c, a) < 0$, and 
    the vertices $a, b, c$ form a negative weight cycle in $G$.
\end{proof}

\section{The All-Alice Case} \label{sec:all-alice}
In this section, we develop an algorithm that computes the minimum sufficient energies of all vertices for game graphs controlled by Alice.
In particular, we obtain the following result.

\ThmAllAlice*

The idea is to use the algorithm of Theorem~\ref{th:apnpW} for the All-Pairs Nonnegative Prefix Paths problem. Hélouët, Markey, and Raha~\cite{HMR19} provide a relevant reduction
from the problem of whether zero energy suffices to the problem of whether there exists a nonnegative prefix path.
Hence, one idea would be to apply this reduction and run the algorithm of Theorem~\ref{th:apnpW}.
Unfortunately this reduction affects the weights, and the maximum weight of the new instance can be as big as $mW$, which
in turn affects the running time of the algorithm.

To that end, we present another way to use the All-Pairs Nonnegative Prefix Paths algorithm of Theorem~\ref{th:apnpW} without affecting the maximum weight of the graph.
The algorithm consists of two phases. In the first phase, we detect all the vertices such that initial zero energy suffices,
and in the second phase we compute the minimum sufficient energy for the rest of the vertices.

In the first phase of the algorithm, initially we run the All-Pairs Nonnegative Prefix Paths algorithm of Theorem~\ref{th:apnpW} on the game graph $G = (V, E, w)$.
Hence, we retrieve the information of whether there exists a nonnegative prefix path from a vertex $u$ to a vertex $v$,
for any two vertices $u, v \in V \times V$. Then for each vertex $v \in V$, we check whether there is a vertex $u$ (including $v$) such that
there exists a nonnegative prefix path from $v$ to $u$ and from $u$ to $u$. If this is the case, then we add this vertex to a set $Z$.
The next lemma shows that the set $Z$ is actually the set of all vertices such that initial energy zero suffices.
%\tijn{recall or refer to def of $e^*$}

\begin{lemma}\label{lm:Alice0}
    The set $Z$ is the same as the set $\{v \in V: e^*(v)=0\}$, and is computed in $\tilde{O}(n^\omega W^\omega)$ time.
\end{lemma}
\begin{proof}
    Suppose that the algorithm adds a vertex $v$ to $Z$. Then, there must be a vertex $u$ (possibly $u = v$) such that there exists
    a nonnegative prefix path from $v$ to $u$ and from $u$ to $u$. By merging then these two paths, and by definition
    of minimum sufficient energy, we can conclude that $e^*(v) = 0$.

    Suppose now that the minimum sufficient energy of a vertex $v \in V$ is zero (i.e., $e^*(v) = 0$). By definition of minimum sufficient
    energy, there must exist a nonnegative prefix lasso $P$ which contains a nonnegative cycle $C$. Also by Lemma~\ref{lem:nonegcyclehasnonnegpref}, there is
    a vertex $u \in C$ in the cycle, such that there exists a nonnegative prefix path from $u$ to itself. As a result, the algorithm finds these vertices $v$ and $u$
    and adds $v$ to $Z$.

    The running time of this process is dominated by the running time of the All-Pairs Nonnegative Prefix Paths algorithm, which is $\tilde{O}(n^\omega W^\omega)$
    based on Theorem~\ref{th:apnpW}.
\end{proof}

The set $Z$ can be seen as the set of possible vertices to `end' in. Any optimal strategy would still have to define how to move from such a vertex $v\in Z$,
but since we know that $e^*(v)=0$, there has to be a path such that from this vertex no initial energy is necessary.
So the goal of the second phase, is to find for each vertex $v \in V \setminus Z$ the best way to hit a vertex in $Z$. 
The following lemma shows that this comes down to a shortest path computation. Brim and Chaloupka \cite{brim2012using} use a similar
idea inside their subroutine for the Mean-Payoff games.

\begin{restatable}{lemma}{LmAlice} \label{lm:Alicerem}
    Given a game graph $G = (V,E,w)$ where all vertices belong to Alice and the set $Z:=\{v\in V: e^*(v)=0\}$ is known, we can compute the remaining minimum sufficient energies through a single SSSP computation in $G$.   
\end{restatable}
Recall that an SSSP computation is the execution of a single-source shortest paths (SSSP) algorithm.
For the proof we refer to the Appendix~\ref{app:missingproofs}. Together, Lemma~\ref{lm:Alice0} and Lemma~\ref{lm:Alicerem} prove Theorem~\ref{thm:allalice}, by using also
the fact that we can compute SSSP deterministically in $\tilde O(n^\omega W)$ time~\cite{Sankowski05,YusterZ05}. 

\section{The All-Bob Case} \label{sec:all-bob}

In this section, we restrict ourselves to the case where all vertices belong to Bob. We show that this special case admits a near-linear time algorithm, by essentially reducing the problem to detecting negative cycles and computing shortest paths. We obtain the following result. 

\ThmAllBob*
\begin{proof}
We split the algorithm and proof in two parts, depending on who wins the game in a particular vertex. The first part of the algorithm consists of identifying the vertices with infinite energy (namely,
the vertices where Bob wins), and the second part consists of calculating the finite energies of the remaining vertices (namely, the vertices where Bob loses).

\textbf{Vertices where Bob wins.}
First, we identify the vertices where Bob wins, i.e., the vertices~$v$ with $e^*(v)=\infty$. Hereto, we decompose $G$ in to strongly connected components $C_1, \dots, C_r$, for some $r\geq 1$. On each $C_i$, we run a negative cycle detection algorithm. If there is a negative cycle, we set $e(v)=\infty$ for all $v\in C_i$. Next we find the vertices that can reach these cycles. Let $A:=\{v\in V : e(v)=\infty\}$ be the union of the 
strongly connected components with a negative cycle. Then from $A$ we run an inward reachability algorithm (e.g., DFS, BFS) towards each vertex $v$ and if there is a path from $v$ to $A$, we set $e(v)=\infty$.
In the correctness proof, we show that $e(v) = \infty$ if and only if Bob wins at $v$.

\textit{Correctness.} For any vertex~$v\in V$, Bob wins if and only if there is a path from $v$ to a negative cycle. Let $v$ be a vertex where Bob wins, and let $C^{(v)}$ be the negative cycle reachable from $v$. If $v$ belongs to the strongly connected component of $C^{(v)}$, then our algorithm outputs $e(v)=\infty$. If $v$ belongs to a different connected component, then the path to the negative cycle is detected in the inward reachability algorithm and we also output $e(v)=\infty$.

Suppose we output $e(v)=\infty$. If we do this because $v$ belongs to a strongly connected in which we detected a negative cycle, then clearly there is path from $v$ to the negative cycle, and hence Bob wins at $v$. If we set $e(v)=\infty$ because there is a path from $v$ to $A$, then there is a path from $v$ towards a strongly connected component containing a negative cycle, and hence to a negative cycle itself. So again Bob wins at $v$. 

\textit{Running time.} We can decompose $G$ in to strongly connected components in $O(m)$ time~\cite{Tarjan72}. On each connected component $C_i$, we can detect whether there is a negative cycle in the graph in $O(|E(C_i)|\log^2 n\log nW\log\log n)$ time w.h.p.~\cite{bringmann2023negative}, thus the total time is $ O(m\log^2n\log nW\log\log n)$ w.h.p. The inward reachability algorithm can be implemented by a simple DFS or BFS in $O(m)$ time. Hence in total we obtain w.h.p\ a running time of $O(m\log^2 n\log nW\log\log n)$ for this part.

\textbf{Vertices where Bob loses.}
Second, we compute the correct value for the vertices where Bob loses, i.e., the vertices~$v$ with $e(v)<\infty$. Note that for this part we can restrict ourselves to the subgraph where we omit all vertices with $e(v)=\infty$. We 
also add a new sink vertex $t$ to the graph, and for every $v \in V$ we insert an edge $(v, t)$ with $w(v, t) = 0$. Now for each vertex~$v$, we compute the minimum distance $d(v,t)$ from $v$ to $t$, and  we set $e(v) = \max\{-d(v,t), 0\}$. In the correctness proof, we show that $e^*(v) = e(v)$ for each $v \in V$ with $e(v) < \infty$.

\textit{Correctness.}
%We write $d(v,u)$ for the minimum distance from $v$ to $u$, that is, the length of the shortest path. 
Consider now a vertex $v$ such that $e(v) < \infty$.
First we show that $e^*(v) \geq e(v)$. Let $u$ be the last vertex (excluding $t$ itself) on the shortest path from $v$ to $t$, and $P_{v, u}$
be the corresponding prefix from $v$ to $u$. Then Bob can choose to move along the path $P_{v, u}$ forcing Alice to use at least $\max\{-w(P_{v, u}), 0\}$
initial energy. As $d(v,t) = w(P_{v, u}) + w(u, t) = w(P_{v, u}) + 0 = w(P_{v, u})$, we conclude that Alice needs at
least $\max\{-d(v,t), 0\} = e(v)$ initial energy.

It remains to show $e^*(v) \leq e(v)$. Since there are no negative cycles, by definition we have that $e^*(v) = \max\{-\min_{u \in V} w(P_u), 0\}$,
where the minimization is over all the simple paths from $v$ to $u$. Also for all $u \in V$, it holds that $d(v, u) \leq w(P_u)$ and $d(v,t) \leq d(v, u) + w(u, t)
= d(v, u) + 0 = d(v, u)$.
Thus we get that $e^*(v) = \max\{-\min_{u \in V} w(P_u), 0\} \leq \max\{-\min_{u \in V} d(v,t), 0\} = \max\{-d(v,t), 0\} = e(v)$.

% \antonis{-----------------------------------------}
% Let $v\in V$ be a vertex with $d(v,t)=0$. That means that for every vertex $u\in V$ reachable by $v$\footnote{In the all-Alice cases we looked at a special type of reachability, where the path had to have no negative prefixes. In this case we consider standard reachability, i.e., we allow all paths.}, we have $d(v,u)\geq 0$, since $d(v,u)+0=d(v,u)+w(u,t)\geq d(v,t) =d(v,t) \geq 0$. Since we know $v$ cannot reach a negative cycle, it must reach a non-negative cycle at some point. Each non-negative cycle has at least one vertex~$u$ with $e^*(u)=0$\tijn{prove this in a lemma?}, and since $d(v,u)=0$, we get $e^*(v)=0=-d(v,t)$. 

% Now let $v$ be any vertex in $V$. Let $u$ be the last vertex (excluding $t$ itself) on the shortest path from $v$ to $t$. Then we have $d(v,u)=d(v,t)$ and $\delta(u)=0$. Clearly the path from $v$ to $u$, and then the path from $u$ towards its positive cycle is a valid strategy for Bob, hence $e^*(v) \geq -d(v,t)$. Now let $\tau^*$ be an optimal strategy for Bob. Then, since there are no negative cycles, $\tau^*$ send $v$ towards a non-negative cycles. In particular, $\tau^*$ sends $v$ to a vertex~$u$ with $e^*(u)=0=\delta(u)$. Since $d(v,t)$ is the shortest distance from $v$ to $t$, we obtain $d(v,t) \leq d(v,u)+w(u,t)=d(v,u)=e^*(v)$. So we obtain $e^*(v)=d(v,t)$ for all $v$. 

\textit{Running time.}
    To compute the shortest paths from $v$ to $t$, we flip the direction of all the edges and we compute the minimum distances from $t$ to $v$ in the new graph. This clearly corresponds to the minimum distances from $v$ to $t$ in the original graph. Since this computation is the negative weight single source shortest path problem, it can be done in $O(m\log^2 n\log nW\log\log n)$ time w.h.p.~\cite{bringmann2023negative}.
\end{proof}

% \begin{lemma} \label{lem:pos_cycle_exist_zero}
    
% \end{lemma}

\section{Game Graphs Without Negative Cycles}\label{sc:nonegcycles}
In this section, we provide an $O(mn)$ time algorithm for the special case where the game graph has no negative cycles. We do this in three steps: first, we introduce a finite duration energy game, where a token is passed for $i$ rounds. The goal is to compute for each vertex, the minimum initial energy that Alice needs in order to keep the energy nonnegative for those $i$ rounds. Second, we provide an algorithm that computes this value in $O(mi)$ time. Finally, we show that for graphs with no negative cycles, it suffices to find this minimum initial energy for a game of $n$ rounds.

%We analyze a version of the `value iteration algorithm'\tijn{should explain what that means}, and show that $n$ iterations suffice for such graphs. In total this gives an $O(mn)$ time algorithm. 
%For the analysis of the value iteration algorithm in this section, we make use of the concept of a game lasting $i$ rounds.

\subsection{Finite Duration Games}
We introduce a version of the energy game that lasts $i$ rounds. We define strategies and energy functions analogous to the infinite duration game, as in Section~\ref{sc:prelim}. A strategy for Alice is a function $\sigma_i: V^*V_A \rightarrow V$, such that
for all finite paths $u_0u_1\cdots u_j$ with $j < i$ and $u_j \in V_A$, we have that
$\sigma_i(u_0u_1\cdots u_j) = v$ for some edge $(u_j, v) \in E$. Similarly we define a strategy $\tau_i$ for Bob by replacing
$V_A$ with $V_B$. A path $u_0u_1\cdots u_j$ of length $j$ is consistent with respect to strategies $\sigma_i$ and $\tau_i$, if $\sigma_i(u_0u_1\cdots u_k) = u_{k+1}$ for all $u_k \in V_A$ and $\tau_i(u_0u_1\cdots u_k) = u_{k+1}$ for all $u_k \in V_B$, where $0 \leq k < j \leq i$.
The minimum sufficient energy at a vertex $u$ corresponding to strategies $\sigma_i$ and $\tau_i$ is defined
as $e_{\sigma_i, \tau_i}(u) := \max\{-\min w(P), 0\}$, where the minimization is over all the consistent paths $P$ with respect to $\sigma_i$
and $\tau_i$ of length at most $i$ originating at $u$. %\sebastian{Just a comment: Instead of consistent strategies we could also talk about prefixes directly} 
The minimum sufficient energy at a vertex $u$ is defined as follows:
$$e^*_i(u) := \min_{\sigma_i} \max_{\tau_i} e_{\sigma_i, \tau_i}(u),$$
where we minimize over all strategies $\sigma_i$ for Alice and maximize over all strategies $\tau_i$ for Bob. As for the infinite duration game, we know by Martin's determinacy theorem~\cite{Martin75} that $ \min_{\sigma_i} \max_{\tau_i} e_{\sigma_i, \tau_i}(u)= \max_{\tau_i}\min_{\sigma_i} e_{\sigma_i, \tau_i}(u)$. Now we define \emph{optimal strategies} as follows. 
A strategy $\sigma^*_i$ is an optimal strategy for Alice at a vertex $u$, if for any strategy $\tau_i$ for Bob it holds that $e_{\sigma^*_i, \tau_i}(u) \leq e^*_i(u)$.
Likewise a strategy $\tau^*_i$ is an optimal strategy for Bob at a vertex $u$, if for any strategy $\sigma_i$ for Alice it holds that $e_{\sigma_i, \tau^*_i}(u) \geq e^*_i(u)$.
A value $e(u)$ is a sufficient energy at a vertex $u$, if there exists a strategy $\sigma_i$ such that for any strategy $\tau_i$, it holds that
$e_{\sigma_i, \tau_i}(u) \leq e(u)$. In this case, observe that the following is true:
$$e^*_i(u) = \max_{\tau_i} e_{\sigma^*_i, \tau_i}(u) \leq \max_{\tau_i} e_{\sigma_i, \tau_i}(u) \leq e(u).$$

% The following lemma shows that the history of a path does not matter for the best next move. The only thing that matters is how many turns are left in the game. 

% \begin{lemma} \label{lem:strat_prop}
%     Consider a game of $i$ rounds, with optimal strategies $\sigma_i^*$ and $\tau_i^*$ for Alice and Bob respectively. Let $u_0u_1\cdots u_j$ be a finite path, consistent with $\sigma_i^*$ and $\tau_i^*$, where $j < i$ and $u_j \in V_A$. Then there exists an
%     optimal strategy $\sigma^*_{i-j}$ such that $\sigma^*_i(u_0u_1\cdots u_j) = \sigma^*_{i-j}(u_j)$.
%     Similarly if $u_j \in V_B$, there exists an optimal strategy $\tau^*_{i-j}$ such that $\tau^*_i(u_0u_1\cdots u_j) = \tau^*_{i-j}(u_j)$.
% \end{lemma}

Next, we show the following lemma about the minimum energy function, a similar version has also been used for the infinite duration game in~\cite{BrimCDGR11} and~\cite{ChatterjeeHKN14}. For the proof, see Appendix~\ref{app:missingproofs}.%The main advantage of such a characterization is that we can locally check whether an arbitrary energy function is a sufficient energy.

\begin{restatable}{lemma}{LmOptiChar} \label{lem:optiprop}
    Given a game of $i$ rounds and a vertex $u \in V$, the energy $e^*_i(u)$ satisfies the following properties:
    \begin{align} 
        \text{if } u \in V_A \text{ then } &\exists v \in N^+(u): e^*_i(u) + w(u, v) \geq e^*_{i-1}(v) \label{eq:Alice}\\
        \text{if } u \in V_B \text{ then } &\forall v \in N^+(u): e^*_i(u) + w(u, v) \geq e^*_{i-1}(v) \label{eq:Bob}
    \end{align}
\end{restatable}

\subsection{A Value Iteration Algorithm for Finite Duration Games}
In this section, we present Algorithm~\ref{alg:valiteralg}, a value iteration algorithm for a game lasting $i$~rounds that computes for each vertex $u \in V$ the value $e^*_i(u)$.
%If we know that $M$ is an upper bound on the number of iterations necessary. Using the trivial upper bound $M=n^2W$, we obtain an $O(mn^2W)$ time algorithm for the general case. 
We note that Algorithm~\ref{alg:valiteralg} consists of $i$ steps, where at every step each edge is scanned at most once. Clearly this means the algorithm takes $O(mi)$ time.

\begin{algorithm}
\caption{Value iteration algorithm for an $i$-round game}\label{alg:valiteralg}
\setcounter{AlgoLine}{0}
\SetAlgoLined
\SetKwComment{Comment}{/* }{ */}
\SetKwInput{KwInput}{Input}
\SetKwInput{KwOutput}{Output}

\KwInput{A game graph $G = (V, E, w, \langle V_A, V_B \rangle)$, a number of iterations $i$}
\KwOutput{The minimum sufficient energy $e_i(u)$ of each $u \in V$, in order to play the game for $i$ rounds}

$\forall u \in V: e_0(u) \gets 0$

\For{$j = 1\; \KwTo \;i$} {
    \ForEach{$u \in V$} {
        \If {$u \in V_A$} {
            $e_j(u) \gets \max\{\min_{(u, v) \in E} \{e_{j-1}(v) - w(u, v)\}, 0\}$ \phantomsection \label{line:valiterAlice}
        }
        
        \If{$u \in V_B$} {
            $e_j(u) \gets \max\{\max_{(u, v) \in E} \{e_{j-1}(v) - w(u, v)\}, 0\}$ \phantomsection \label{line:valiterBob}
        }
    }
}

\textbf{return} $e_i$
\end{algorithm}

\begin{lemma} \label{lem:invalg}
    Let $e_i(\cdot)$ be the function returned by Algorithm~\ref{alg:valiteralg}, then $e_i(u) = e^*_i(u)$ for all $u \in V$.
\end{lemma}
\begin{proof}
    We prove the claim by induction on $i$, which is both the number of steps of the algorithm and the duration of the game.
    
    \emph{Base case:} For $i = 0$ steps, the algorithm sets for each $u \in V: e_0(u) = 0 = e^*_0(u)$.
        
    \emph{Inductive Step:} We assume that after $i - 1$ steps $e_{i-1}(u) = e^*_{i-1}(u)$, and we prove that after $i$ steps $e_i(u) = e^*_i(u)$
        as well. We first show that $e_i(u) \geq e^*_i(u)$.
        
        Consider the case that $u \in V_A$. Let $v'$ be the neighbor that minimizes the relation in the $i_{th}$ step in Line~\ref{line:valiterAlice}.
        Then it holds that $e_i(u) + w(u, v') \geq e_{i-1}(v')$. Using the edge $(u, v')$
        with initial energy $e_i(u)$, Alice can move to $v'$ with remaining energy at least $e_{i-1}(v')$.
        By the inductive hypothesis it holds that $e_{i-1}(v') = e^*_{i-1}(v')$, so there exists an optimal strategy $\sigma^*_{i-1}$ such that
        for any strategy $\tau_{i-1}$, we have that $e_{\sigma^*_{i-1}, \tau_{i-1}}(v') \leq e_{i-1}(v')$.
        Define the strategy $\sigma_i$ in the following way: $\forall x \in V^*V_A: \sigma_i(ux) = \sigma^*_{i-1}(x)$ and $\sigma_i(u) = v'$.
        Then we get a strategy $\sigma_i$ such that for any strategy $\tau_i$, it holds that $e_{\sigma_i, \tau_i}(u) \leq e_i(u)$.
        This implies that $e_i(u)$ is a sufficient energy at vertex $u$, and so $e_i(u) \geq e^*_i(u)$.
        
        Consider the case that $u \in V_B$. Due to the $i_{th}$ step in Line~\ref{line:valiterBob}, it holds that $e_i(u) + w(u, v) \geq e_{i-1}(v)$,
        for all $v \in N^+(u)$. Hence for any choice of a neighboring edge $(u, v)$ with initial energy $e_i(u)$, Bob moves to a neighbor $v$ with remaining
        energy at least $e_{i-1}(v)$. By the inductive hypothesis, for all $v \in N^+(u)$ it holds that $e_{i-1}(v) = e^*_{i-1}(v)$, so there exists an 
        optimal strategy $\sigma^*_{i-1}$ such that
        for any strategy $\tau_{i-1}$, we have that $e_{\sigma^*_{i-1}, \tau_{i-1}}(v) \leq e_{i-1}(v)$. Define the strategy $\sigma_i$
        in the following way: $\forall x \in V^*V_A: \sigma_i(ux) = \sigma^*_{i-1}(x)$. Then we get a strategy $\sigma_i$ such that for any strategy $\tau_i$,
        it holds that $e_{\sigma_i, \tau_i}(u) \leq e_i(u)$. This implies that $e_i(u)$ is a sufficient energy at vertex $u$, and so $e_i(u) \geq e^*_i(u)$.
        
        It remains to show that $e_i(u) \leq e^*_i(u)$. Consider the case that $u \in V_A$. If $e_i(u) = 0$ then the claim holds trivially. If $e_i(u) > 0$, then based on
        Line~\ref{line:valiterAlice}, we have that $e_i(u) + w(u, v) \leq e_{i-1}(v)$ for all $v \in N^+(u)$. By Lemma~\ref{lem:optiprop}, 
        there exists $v' \in N^+(u)$ such that $e^*_i(u) + w(u, v') \geq e^*_{i-1}(v')$, which means that:
        $$e_i(u) + w(u, v') \leq e_{i-1}(v') = e^*_{i-1}(v') \leq e^*_i(u) + w(u, v') \; \Rightarrow \; e_i(u) \leq e^*_i(u),$$
        where the equality holds by the inductive hypothesis.
        
        Consider the case that $u \in V_B$. If $e_i(u) = 0$ then the claim holds trivially. Otherwise based on Line~\ref{line:valiterBob}, there exists $v' \in N^+(u)$
        such that $e_i(u) + w(u, v') = e_{i-1}(v')$. By Lemma~\ref{lem:optiprop}, we have that $e^*_i(u) + w(u, v) \geq e^*_{i-1}(v)$
        for all $v \in N^+(u)$, which means that:
        $$e_i(u) + w(u, v') = e_{i-1}(v') = e^*_{i-1}(v') \leq e^*_i(u) + w(u, v') \; \Rightarrow \; e_i(u) \leq e^*_i(u),$$
        where the equality holds by the inductive hypothesis.\qedhere
        %
        %Therefore we can conclude that $e_i(u) = e^*_i(u)$ for all $u \in V$.
\end{proof}

\subsection{No Negative Cycles}
The goal of this section is to show that for graphs with no negative cycles, it holds that $e^*_n(u) = e^*(u)$, for all $u \in V$. 
Hereto, we show in Lemma~\ref{lem:posstrnonneg} that as in the infinite duration game, positional strategies suffice when no negative cycles are present. 
%\antonis{explain what happens with negative cycles.}
In the proof, we use the following alternative characterization of $e_{\sigma_i, \tau_i}(u)$.

Let $\sigma_i$ and $\tau_i$ be strategies for Alice and Bob respectively, and let $u\in V$ be a vertex. Moreover, let $u_0u_1 \cdots u_j$ be the consistent path of length $j$ with respect to $\sigma_i$ and $\tau_i$, where
$u_0=u$. Then given an initial energy $e_{\text{init}}$, the energy level at vertex $u_j$ is equal to the value $e_{\text{init}} + \sum_{k=0}^{j-1} w(u_k, u_{k+1})$. 
We denote $e^*_{\rm{init}}(u)$ for the minimum nonnegative initial energy such that the energy level at each vertex of the corresponding consistent path of length $i$, is nonnegative. 
The following lemma shows that $e_{\sigma_i, \tau_i}(u)=e^*_{\rm{init}}(u)$ (for the proof, see Appendix~\ref{app:missingproofs}).

\begin{restatable}{lemma}{LmMinenLev} \label{lem:minenlev}
    For a vertex $u$ and two fixed strategies $\sigma_i$ and $\tau_i$, let $P$ be the consistent path with respect to $\sigma_i$ and $\tau_i$ of length $i$ originating at $u$.
    Then it holds that $e_{\sigma_i, \tau_i}(u)=e^*_{\rm{init}}(u)$. 
\end{restatable}

Now we are ready to show that positional strategies suffice in graphs without negative cycles. For the proof see Appendix~\ref{app:missingproofs}.

\begin{restatable}{lemma}{LmPositional} \label{lem:posstrnonneg}
    Consider a graph with no negative cycles and a game of $i$ rounds. Then for the minimum sufficient energy $e^*_i(u)$ at a vertex $u \in V$,
    it suffices for both players to play positional strategies.
\end{restatable}

We use this fact to show that a game of $n$ rounds is equivalent to a game of infinite duration for a game graph without negative cycles. 

\begin{lemma} \label{lem:nonnegntoinf}
    Consider a graph with no negative cycles. Then for each vertex $u \in V$, the minimum sufficient energy needed at $u$ for a game of $n$ rounds,
    is equal to the minimum sufficient energy needed at $u$ for a game of infinite rounds. In other words, $e^*_n(u) = e^*_\infty(u) = e^*(u)$
    for all $u \in V$.
\end{lemma} 
%\textbf{Unfortunately the lemma cannot be generalized.
%Consider the graph a->b(-1), b->a(-1), a->c(-W), b->c(-W), c->d(0), d->c(0). The notation u->v(w) means directed edge from u to v with weight w.
%Assume that all the vertices belong to Alice. Then Alice wins the game at every vertex of the graph.
%Consider the optimal play at vertex $a$ assuming that $|W| > n$.
%The optimal play for $n$ rounds at $a$ is to follow the cycle a->b with initial energy equal to $n$.
%The optimal play for infinite rounds at $a$ is to follow the cycle c->d with initial energy equal to $W$.}
\begin{proof}
    %By definition it is true that: 
    %$$e^*_n(u) = \min_{\sigma_n} \max_{\tau_n} \max\{-\min_{|P| \leq n} w(P), 0\} \leq \min_{\sigma} \max_{\tau} \max\{-\min w(Q), 0\} = e^*(u)$$
    %where $\sigma, \tau$ behave as $\sigma_n, \tau_n$ in the first $n$ rounds, and $P, Q$ are the corresponding consistent paths.
    %Hence it remains to show that $e^*_n(u) \geq e^*(u)$.
    Let $\sigma$ and $\tau$ be two arbitrary positional strategies for the infinite duration game. By definition, we have that 
    $e_{\sigma, \tau}(u) = \max\{-\min w(P), 0\}$, where the minimization is over all the consistent paths with respect to $\sigma$ and $\tau$ 
    originating at $u$. Since the graph contains only nonnegative cycles and the strategies are positional, the path that minimizes the relation is a simple path,
    and so, its length is at most $n$. Hence it follows that $e_{\sigma, \tau}(u) = \max\{-\min_{|P| \leq n} w(P), 0\}$. In turn, this is equivalent to using positional strategies
    for a game of $n$ rounds. Hence it holds that $e_{\sigma, \tau}(u) = e_{\sigma_n, \tau_n}(u)$, where $\sigma_n$ and $\tau_n$
    are the strategies $\sigma$ and $\tau$ respectively, restricted to the first $n$ rounds. This implies that
    $e^*(u) = \min_{\sigma_n} \max_{\tau_n} e_{\sigma_n, \tau_n}(u)$, where $\sigma_n$ and $\tau_n$ are positional strategies for a game of $n$ rounds.
    By Lemma~\ref{lem:posstrnonneg}, this equals $e_n^*(u)$ and the claim follows.
    %and so $e^*_n(u) = \min_{\sigma_n} \max_{\tau_n'} e_{\sigma_n, \tau_n'}(u)$. Finally by restricting to positional strategies $\tau_n$ for Bob,
    %we get that $e^*_n(u) \geq \min_{\sigma_n} \max_{\tau_n} e_{\sigma_n, \tau_n}(u) = e^*(u)$, and the claim follows.
\end{proof}

Together, Lemma~\ref{lem:invalg} and Lemma~\ref{lem:nonnegntoinf} prove Theorem~\ref{thm:nonnegcycles}.%, which we restate for convenience. 

%\ThmNoNegCycles*
% \begin{theorem} \label{th:nonneg_cycle_iter}
%     Consider a graph with no negative cycles. Then the running time of the value iteration algorithm is $O(mn)$.
% \end{theorem}
% \begin{proof}
%     By Lemma~\ref{lem:nonneg_n_to_inf} we have that $e^*(u) = e^*_n(u)$. Lemma~\ref{lem:inv_alg} shows that Algorithm~\ref{alg:val_iter_alg} outputs $ e_n(u)=e^*_n(u)$ in $n$ steps, which takes $O(mn)$ time. 
% \end{proof}

 \section{Conclusion}
In this paper, we presented algorithms for three special cases of energy games. For the first and second result, regarding games with a single player, we demonstrated a relation between shortest path computation and energy games. These relations are interesting by themselves since they reveal more about the nature of energy games.

In particular for the all-Alice case, we make the connection to the All-Pairs Nonnegative Prefix Paths problem explicit. We show that under the APSP conjecture there is no $O(n^{3-\epsilon})$ time algorithm for the reachability version of this problem. Since this problem is very similar to the all-Alice case, it might be a hint that also the all-Alice case allows for no $O(n^{3-\epsilon})$ time algorithm.

For the all-Bob case, we essentially solve the problem by giving a near-linear time algorithm. We do this by exploiting the recent progress on the negative-weight single-source shortest paths problem.

For the case with no negative cycles, we introduce a variant of the value iteration algorithm which takes $O(mn)$ time, improving the algorithm of \cite{ChatterjeeHKN14} for this case by a $\log W$ factor. Similar to Bellman-Ford, it seems that directly improving the iteration count of a value iteration algorithm below $n$ is hard -- if not impossible. It would be interesting to see a faster algorithm for graphs without negative cycles, which most probably will (also) use different ideas.

%\subsection*{Acknowledgements}

 \printbibliography[heading=bibintoc] % Make bibliography show up in table of contents

%\listoftodos

\appendix

\section{Missing Proofs}\label{app:missingproofs}
\LmAlice*
\begin{proof} 
    First we construct the graph $G_t=(V_t, E_t)$ by deleting all the outgoing edges of $Z$, and contracting $Z$ to a single target vertex $t$. Now we show that the problem reduces to computing the maximum distance\footnote{We define the \emph{maximum distance} between $v$ and $t$ as the maximum weight of any path from $v$ to $t$. Note that in graphs without positive cycles this is either a finite value, or $\infty$ if there is no $s-t$ path.} 
    from each vertex $v$ to $t$ in $G_t$, denoted by $\delta(v,t)$. 
    First note that if $v$ can reach $t$ in $G_t$, then $e^*(v)$ is finite and the answer is obtained from the distance computation. If $v$ cannot reach $t$ in $G_t$, then we must have $e^*(v)=\infty$, because in different case $v$ could reach $t$ in
    $G_t$ (see Lemma~\ref{lem:nonegcyclehasnonnegpref}).
    Thus we assume that $v$ can reach $t$ in $G_t$, and we prove the lemma in two claims. 
    % Also, notice that 
    % if the target vertex $t$ is not reachable from a vertex $v$ in $G_t$ then
    % $e^*(v) = \infty$, since otherwise either $v \in Z$ or
    % $v$ can reach $t$ in $G_t$ (see Lemma~\ref{lem:nonegcyclehasnonnegpref} and Lemma~\ref{lm:Alice0}). 
    % Thus, in the rest we assume that every vertex
    % can reach the target vertex $t$ in $G_t$.
    % \antonis{This is new: check correctness.}
    
    \textbf{Claim 1.} $\delta(v,t)<0$ for all $v\neq t$. \\
    Suppose to the contrary that there exists a vertex $v$ such that $\delta(v,t)\geq 0$, and let $P$ be a path from $v$ to $t$ of weight $\delta(v,t)$ (i.e., $w(P) = \delta(v, t)$). Denote $u$ for the last vertex on $P \setminus \{t\}$ such that
    $\delta(u,t)\geq 0$. Note that such $u$ always exists, since $v$ itself is a valid candidate. Denote $u=u_0, u_1, \dots, u_{k-1},u_k = t$ for the subpath of $P$ from $u$ to $t$. If $\sum_{j=0}^{i-1} w(u_j, u_{j+1}) \leq 0$ for some $i \in \{1, \dots, k-1\}$, then 
    $\delta(u_i,t) \geq \sum_{j=i}^{k-1} w(u_j, u_{j+1}) = \delta(u,t) - \sum_{j=0}^{i-1} w(u_j, u_{j+1}) \geq \delta(u,t) \geq 0$, which contradicts the fact that $u$ is the last vertex on $P \setminus \{t\}$
    with $\delta(u,t)\geq 0$. Hence we have that $\sum_{j=0}^{i-1} w(u_j, u_{j+1}) > 0$ for all $i \in \{1, \dots, k-1\}$ and together with $\delta(u,t) \geq 0$, we can conclude that the path $u=u_0, u_1, \dots, u_{k-1},u_k = t$
    has only nonnegative prefixes.
    This implies that $e^*(u) = 0$, and based on Lemma~\ref{lm:Alice0} 
    we have that $u \in Z$. However, this is a contradiction on the fact that $u \neq t$, and so, we conclude that $\delta(v,t)<0$ for all $v\neq t$. 
    
    \textbf{Claim 2.} $e^*(v)=-\delta(v,t)$.\\
    By the definition of $e^*(\cdot)$ there is a path of weight at least $-e^*(v)$ from $v$ to $t$, so we have $\delta(v,t) \geq -e^*(v)$, or equivalently $e^*(v)\geq -\delta(v,t)$. 
    For the other inequality, suppose for contradiction that $e^*(v) > -\delta(v,t)$ and let $P$ be a path from $v$ to $t$ of weight $\delta(v,t)$. By Claim 1, we have that $-\delta(v,t) > 0$, hence $e^*(v)>0$. Now by the definition of $e^*(\cdot)$ and $t$, since $e^*(v)$ does not use $P$ as a certificate,    
    there must exist a prefix
    $u_0, \dots, u_i$ of $P \setminus \{t\}$ such that $\sum_{j=0}^{i-1} w(u_j, u_{j+1}) < \delta(v,t)$. However, we know that the following is true:
    $$\delta(v, t) = \sum_{j=0}^{i-1} w(u_j, u_{j+1}) + w\Bigl(P \setminus \bigcup_{j=0}^{i-1}u_j\Bigr) \leq \sum_{j=0}^{i-1} w(u_j, u_{j+1}) + \delta(u_i,t) < \sum_{j=0}^{i-1} w(u_j, u_{j+1}),$$ where the last inequality comes from Claim 1.
    This yields a contradiction, and so the claim follows.
    
    \textbf{Algorithm.} We note that computing maximum distances from $v$ to $t$ in $G_t=(V_t,E_t)$, is the same as computing maximum distances from $t$ to $v$ in $G'_t=(V_t,E'_t)$, where $E'_t = \{(u,v) : (v,u)\in E_t\}$ is the graph with all edges reversed. Computing maximum distances in $(G'_t,w)$ is equivalent to computing minimum distances in $(G'_t,-w)$ and negating the output.
    
    Finally, remark that there are no positive cycles
    in $(G_t, w)$ or $(G'_t, w)$. This can be justified as follows. By Lemma~\ref{lem:nonegcyclehasnonnegpref},
    in every positive cycle in $G_t$ there is a vertex $u$ such that $e^*(u) = 0$,
    which in turn by Lemma~\ref{lm:Alice0} implies that $u \in Z$.
    Hence, as we construct $G_t$ by deleting all the outgoing edges of $Z$ and contracting $Z$ to $t$,
    there are no negative cycles in $(G'_t,-w)$.
    Therefore, we conclude that we can compute $\delta(\cdot,t)$, and hence $e^*(\cdot)$, by computing single-source shortest paths on $(G'_t,-w)$. 
\end{proof}

\LmOptiChar*
\begin{proof}
Throughout this lemma we write $P_x$ for a path consistent with the relevant strategies starting at $x$. 
Consider a vertex $u\in V_A$. First of all, we notice that Property~\ref{eq:Alice} is equivalent to
$$e^*_i(u) \geq \min_{v\in N^+(u)} e^*_{i-1}(v) - w(u,v).$$
By definition we have that:
\begin{align*}
    e^*_i(u) &= \min_{\sigma_i}\max_{\tau_i} \max\{-\min_{|P_u|\leq i} w(P_u),0\} \\
    &= \min_{v\in N^+(u)} \min_{\sigma_{i-1}}\max_{\tau_{i-1}} \max\{-w(u, v) - \min_{0 \leq |P_v|\leq i-1} w(P_v),0\},
\end{align*}
where the last equality holds by writing out the different options for a strategy of Alice at $u$. Notice that we allow $P_v$ in the minimization to be empty,
which means that $-\min_{0 \leq |P_v|\leq i-1} w(P_v)\geq 0$. Now we distinguish two cases: 
If $-w(u,v)-\min_{0\leq|P_v|\leq i-1} w(P_v)\geq 0$, then we get 
\begin{align*}
    \max\{-w(u,v)-\min_{0\leq|P_v|\leq i-1} w(P_v),0\} &= -w(u,v)-\min_{0\leq|P_v|\leq i-1} w(P_v) \\
    &= -w(u,v)+\max\{-\min_{0\leq|P_v|\leq i-1} w(P_v),0\}.
\end{align*}
If $-w(u,v)-\min_{0\leq|P_v|\leq i-1} w(P_v)\leq 0$, then we get 
\begin{align*}
    \max\{-w(u,v)-\min_{0\leq|P_v|\leq i-1} w(P_v),0\} = 0 \geq -w(u,v)+\max\{-\min_{0\leq|P_v|\leq i-1} w(P_v),0\}.
\end{align*}
Therefore it follows that
\begin{align*}
    e^*_i(u) &\geq \min_{v\in N^+(u)} \min_{\sigma_{i-1}}\max_{\tau_{i-1}} \max\{-\min_{0\leq|P_v|\leq i-1} w(P_v),0\}-w(u,v)\\
    &= \min_{v\in N^+(u)} e^*_{i-1}(v)-w(u,v).
\end{align*}
Consider now a vertex $u\in V_B$. Similarly, we notice that Property~\ref{eq:Bob} is equivalent to
$$ e^*_i(u) \geq \max_{v\in N^+(u)} e^*_{i-1}(v)-w(u,v).$$
Following the same argument as before, we obtain
\begin{align*}
    e^*_i(u) &= \max_{\tau_i}\min_{\sigma_i} \max\{-\min_{|P_u|\leq i} w(P_u),0\} \\
    &= \max_{v\in N^+(u)} \max_{\tau_{i-1}}\min_{\sigma_{i-1}} \max\{-w(u, v) -\min_{0\leq|P_v|\leq i-1} w(P_v),0\}\\
    &\geq \max_{v\in N^+(u)} \max_{\tau_{i-1}}\min_{\sigma_{i-1}} \max\{-\min_{0\leq|P_v|\leq i-1} w(P_v),0\}-w(u,v)\\
    &= \max_{v\in N^+(u)}e^*_{i-1}(v)-w(u,v).\qedhere
\end{align*}
\end{proof}

\LmMinenLev*
\begin{proof}
    To show that $e_{\sigma_i, \tau_i}(u) \geq e^*_{\text{init}}(u)$, observe that with initially energy $e_{\sigma_i, \tau_i}(u)$, the energy level at each vertex of $P$ is nonnegative. 
    If this is not the case, then there must exist a prefix $P'$ of $P$ such that $e_{\sigma_i, \tau_i}(u) + w(P') < 0$, which yields a contradiction on the definition of $e_{\sigma_i, \tau_i}(u)$.
    It remains to show that $e_{\sigma_i, \tau_i}(u) \leq e^*_{\text{init}}(u)$. If $e_{\sigma_i, \tau_i}(u) = 0$, then trivially $e_{\sigma_i, \tau_i}(u) \leq e^*_{\text{init}}(u)$. Otherwise let $P'$ be the prefix of $P$ 
    that minimizes the value of $e_{\sigma_i, \tau_i}(u)$. Since for the value $e^*_{\text{init}}(u)$ it holds that $e^*_{\text{init}}(u) + w(H) \geq 0$ for every prefix $H$ of $P$, this implies that $e^*_{\text{init}} (u)
    \geq -w(P') = e_{\sigma_i, \tau_i}(u)$, and so the claim follows.
\end{proof}

\LmPositional*
\begin{proof}
    At first, we show the claim for Alice.
    Consider a vertex $u \in V$ and let $\sigma^*_i$ and $\tau^*_i$ be two optimal strategies, namely $e^*_i(u) = e_{\sigma^*_i, \tau^*_i}(u)$. Define $\sigma_i$
    to be the corresponding positional strategy with respect to $\sigma^*_i$ in the following way. Along the consistent path
    $u=u_0\cdots u_j \cdots u_i$ with respect to $\sigma^*_i$ and $\tau^*_i$, whenever we are in a vertex $u_j \in V_A$ for the first time we set
    $\sigma_i(u_j) = \sigma^*_i(u_0\cdots u_j)$. The claim is that $\sigma_i$ is also optimal, that is, $e_{\sigma_i, \tau_i}(u) \leq e^*_i(u)$
    for any strategy $\tau_i$. Along the consistent path with respect to $\sigma^*_i$ and $\tau^*_i$, let $v$ be a vertex that is visited twice,
    $P_1$ be the subpath from $u_0$ to $v$ when $v$ is visited for the first time, and $P_2$ be the path from $u_0$ to $v$ when $v$ is visited again.
    Since the graph contains only nonnegative cycles, the weight of $P_2$ cannot be smaller than the weight of $P_1$ and the energy level
    at $v$ in $P_2$ is not smaller than the energy level at $v$ in $P_1$. Together with Lemma~\ref{lem:minenlev}, this implies that with initial energy $e_{\sigma^*_i, \tau^*_i}(u)$, Alice can safely 
    move from $v$ to $\sigma_i(v)$ the next time as well.
    Hence with the strategies $\sigma_i$, $\tau^*_i$ and with initial energy $e_{\sigma^*_i, \tau^*_i}(u)$, 
    the energy level at each vertex of the corresponding consistent path (the path with respect to $\sigma_i$ and $\tau^*_i$ of length $i$ originating at $u$)
    is nonnegative. Therefore by Lemma~\ref{lem:minenlev}, it holds that $e_{\sigma_i, \tau^*_i}(u) \leq e_{\sigma^*_i, \tau^*_i}(u)$, and the claim follows.
    
    We now show the claim for Bob as well. Consider a vertex $u \in V$ and let $\sigma^*_i$ and $\tau^*_i$ be two optimal strategies, namely $e^*_i(u) = e_{\sigma^*_i, \tau^*_i}(u)$.
    By definition $e_{\sigma^*_i, \tau^*_i}(u) = \max\{-\min_j w(P_j), 0\}$, where $P_j$ is the consistent path with respect to $\sigma^*_i$ and $\tau_i$ of length $j \leq i$ originating at $u$.
    We denote by $P$ the path that minimizes the relation, and we distinguish two cases. The first case is when Bob's strategy is positional at each vertex of $P$, and trivially we can replace
    $\tau^*_i$ with the corresponding positional strategy. For the second case, suppose that there exists a vertex $v \in P$ where Bob chooses a different neighbor $v'$ than the one selected the first time. 
    Since the graph contains only nonnegative cycles, if Bob would have moved to $v'$ the first time, the weight of $P$ without the cycle would not exceed the weight of $P$. Therefore we can replace
    $\tau^*_i$ with a positional strategy in this case as well.
\end{proof}

\end{document}